\documentclass[a4paper]{article}

\usepackage{booktabs} 

\usepackage{amsmath,amssymb,amsthm}
\usepackage{graphicx}
\usepackage{xcolor}
\usepackage{tikz}
\usetikzlibrary{arrows,automata}
\usepackage{thmtools}
\usepackage{comment}

\newcommand{\Gurus}{\mathtt{Gu}} 
\newcommand{\GuruOf}{\mathtt{gu}}
\newcommand{\AttBy}{\mathtt{Att}}
\newcommand{\SN}{\mathcal{G}_{\mathtt{SN}}}
\newcommand{\VoterSet}{\mathcal{N}}
\newcommand{\EdgeSet}{\mathcal{E}}
\newcommand{\AbstSet}{\mathcal{A}}
\newcommand{\NeighbSet}{\mathtt{Nb}}

\newcommand{\Acc}{\mathtt{Acc}} 
\newcommand{\Score}{\mathtt{rk}}
\newcommand{\VotingPower}{\mathtt{vp}}
\newcommand{\GP}{G_P} 
\newcommand{\AP}{A_P}

\newcommand{\Tree}{\mathcal{T}}
\newcommand{\parent}{p}
\newcommand{\Child}{Ch}
\newcommand{\Subtree}[1]{\Tree_{#1}}
\newcommand{\Tr}{\Subtree{r}}
\newcommand{\Tu}{\Subtree{u}}
\newcommand{\Trsansr}{\Tr \setminus \{ r \}}
\newcommand{\OrAbst}{\cup \{0\}}
\newcommand{\Label}{\operatorname{Label}}

\newtheorem{obs}{Observation}
\newtheorem{example}{Example}
\newtheorem{proposition}{Proposition}
\newtheorem{theorem}{Theorem}
\newtheorem{corollary}{Corollary}

\title{The Convergence of Iterative Delegations in Liquid Democracy in a Social Network}

\author{Bruno Escoffier$^{1,2}$, Hugo Gilbert$^3$, Ad\`ele Pass-Lanneau$^{1,4}$}
\date{}

\begin{document}
\maketitle

\begin{center}
$^1$ Sorbonne Universit\'e, CNRS, LIP6 UMR 7606, 

4 place Jussieu, 75005 Paris, France.

$^{2}$ Institut Universitaire de France,

$^3$ Gran Sasso Science Institute,

Viale Francesco Crispi 7, 67100, L'Aquila, Italy

$^4$ EDF R\&D, 7 boulevard Gaspard Monge, 91120 Palaiseau.

\end{center}

\begin{abstract}
Liquid democracy is a collective decision making paradigm which lies between direct and representative democracy. One of its main features is that voters can delegate their votes in a transitive manner such that: A delegates to B and B delegates to C leads to A indirectly delegates to C. These delegations can be effectively empowered by implementing liquid democracy in a social network, so that voters can delegate their votes to any of their neighbors in the network. However, it is uncertain that such a delegation process will lead to a stable state where all voters are satisfied with the people representing them. We study the stability (w.r.t. voters preferences) of the delegation process in liquid democracy and model it as a game in which the players are the voters and the strategies are their possible delegations. We answer several questions on the equilibria of this process in any social network or in social networks that correspond to restricted types of graphs.

We show that a Nash-equilibrium may not exist, and that it is even NP-complete to decide whether one exists or not. This holds even if the social network is a complete graph or a bounded degree graph. We further show that this existence problem is W[1]-hard w.r.t. the treewidth of the social network. Besides these hardness results, we demonstrate that an equilibrium always exists whatever the preferences of the voters iff the social network is a tree. We design a dynamic programming procedure to determine some desirable equilibria (e.g., minimizing the dissatisfaction of the voters) in polynomial time for tree social networks.  Lastly, we study the convergence of delegation dynamics. Unfortunately, when an equilibrium exists, we show that a best response dynamics may not converge, even if the social network is a path or a complete graph.   
\end{abstract}


\section{Introduction}
\label{sec:intro}
 
\emph{Liquid Democracy} (LD) is a voting paradigm which offers a middle-ground between direct and representative democracy. One of its main features is the concept of \emph{transitive delegations}, i.e., each voter can delegate her vote to some other  voter, called representative or proxy, which can in turn delegate her vote and the ones that have been delegated to her to another voter. Consequently, a voter who decides to vote has a voting weight corresponding to the number of people she represents, i.e., herself and the voters who directly or indirectly delegated to her. This voter is called the \emph{guru} of the people she represents. 
This approach has been advocated recently by many political parties as the German's Pirate party or the Sweden's Demoex party and is implemented in several tools as the online platforms LiquidFeedback \cite{behrens2014principles} and liquid.us or GoogleVotes \cite{hardt2015google}. 
One main advantage of this framework is its flexibility, as it enables voters to vote directly for issues on which they feel both concerned and expert and to delegate for others. 
In this way, LD provides a middle-ground between direct democracy, which is strongly democratic but which is likely to yield high abstention rates or uninformed votes, and representative democracy which is more practical but less democratic \cite{green2015direct,wiersma2013transitive}. Importantly, LD can be conveniently used in a Social Network (SN), where natural delegates are connected individuals. These choices of delegates are also desirable as they ensure that delegations rely on a foundation of trust. For these reasons, several works studying LD in the context of an SN enforce the constraint that voters may only delegate directly to voters that are connected to them~\cite{boldi2009voting,kahng2018liquid,bloembergen2018rational}.

\paragraph{Aim of this paper.}  
In this work, we tackle the problem of the stability of the delegation process in the LD setting. Indeed, it is likely that the preferences of voters over possible gurus will be motivated by different criteria and possibly contrary opinions. Hence, the iterative process where each voter chooses her delegate may end up in an unstable situation, i.e., a situation in which some voters would change their delegations. A striking example to illustrate this point is to consider an election where the voters could be positioned on the real line in a way that represents their right-wing left-wing political identity. If voters are ideologically close enough, each voter, starting from the left-side, could agree to delegate to her closest neighbor on her right. By transitivity, this would lead to all voters, including the extreme-left voters, having an extreme-right voter for guru. These unstable situations raise the questions: ``Under what conditions do the iterative delegations of the voters always reach an equilibrium? Does such an equilibrium even exist? Can we determine equilibria that are more desirable than others?''. 

We assume that voters are part of an SN represented by an undirected graph so that they can only delegate directly to one of their neighbors. The preferences of voters over possible gurus are given by unrestricted linear preference orders. In this setting, the delegation process yields a game where players are the voters involved in the election and each player seeks to minimize the rank of her guru in her preference order. We answer the questions raised above with a special emphasize on the SN. 

\paragraph{Our results.} We first show that in a complete SN, the problem of determining if an equilibrium exists is equivalent to the NP-complete problem of determining if a digraph admits a kernel.  
Then we strengthen this result by giving hardness results on two parameters of the SN, the maximum degree of a node and the treewidth. Secondly, we obtain positive results. Indeed, we show that an equilibrium is guaranteed to exist iff the SN is a tree. We study the case of trees for which we design efficient algorithms for the computation of several desirable equilibria. 
Lastly, we study the convergence of delegation dynamics.
Unfortunately, when an equilibrium exists, we show that a best response dynamics may not converge, even if the SN is a path or is complete.

\section{Related Work}\label{sec:RelatedWork}
The stability of the delegation process is one of the several algorithmic issues raised by LD. These issues have recently raised attention in the AI literature. We review some of these questions here.\\

\textit{Are votes in an LD setting ``more correct''?} The idea underlying LD is that its flexibility should allow each voter to make an informed vote either by voting directly, or by finding a suitable guru. Several works have investigated this claim as the ones of Green-Armytage and Kahng et al. \cite{kahng2018liquid,green2015direct}.
In the former work \cite{green2015direct}, the author proposes a setting of spacial voting in which voters' opinions on different issues of a vote are given by their position on the real line (one position per issue). Moreover, the competence level of each voter is expressed as a random variable which adds noise to the estimates that she has about voters' positions (hers included). Votes and delegations are then prescribed by these noisy estimates and by the competence levels of the voters. Green-Armytage defines the expressive loss of a voter as the squared distance between her vote and her position and proves that transitive delegations decrease on average this loss measure. In the latter work~\cite{kahng2018liquid}, the authors study an election on a binary issue, for which there is one ``correct'' answer. Each voter has a competence level, i.e., a probability of being correct. They further assume that the voters belong to an SN and that each voter only accepts to delegate to her neighbors who are ``sufficiently'' more expert than her. The authors investigate delegation \emph{procedures} which take as input the SN and the voters' competence levels and output the probabilities with which each voter should vote or delegate to her approved neighbors. Their key result consists in showing that no ``local'' procedure (i.e., procedures s.t. the probability distribution for each voter only depends on her neighborhood) can guarantee that LD is, at the same time, never (in large enough graphs) less accurate and sometimes strictly more accurate than direct voting.\\

\textit{How much delegations should a guru get?} This question is twofold. On the one hand, gurus should have incentives to obtain delegations, but on the other hand it would be undesirable if a guru was to become too powerful.  This problem has led to two recent papers \cite{golzfluid,kotsialou2018incentivising}. In  the work of Kotsialou and Riley \cite{kotsialou2018incentivising}, voters have both preferences over candidates and over possible gurus.  Given the preferences over gurus, a \emph{delegation rule function} decides who should get the delegations, and then, a \emph{voting rule function} decides who wins the election given the preferences and the voting power of each guru. Focusing on two delegation rules named \emph{depth first delegation rule} and \emph{breadth first delegation rule}, they show that the latter one guarantees that a guru is always better  of w.r.t. the outcome of the election when receiving a delegation whereas it is not the case for the former one. This shows that incentivising participation (voting and delegating) can be a concern in LD. In the setting of G\"olz et al. \cite{golzfluid}, each voter can decide to vote or to specify multiple delegation options. Then a centralized algorithm should select delegations to minimize the maximum voting power of a voter. The authors give a $(1+\log(n))$-approximation algorithm ($n$ being the number of voters) and show that approximating the problem within a factor $\frac{1}{2}\log_2(n)$ is NP-hard. Lastly, they gave evidence that allowing voters to specify multiple possible delegation options (instead of one) leads to a dramatic decrease of the maximum voting power of a voter.\\

\textit{Are votes in an LD setting rational?} 
In the work of Christoff and Grossi \cite{christoffbinary}, the authors study the potential loss of a rationality constraint when voters should vote on different issues that are logically linked and for which they delegate to different proxies. In the continuation of this work, Brill and Talmon~\cite{Brill2018liquid} considered an LD framework in which each voter should provide a linear order over possible candidates. To do so each voter may delegate different binary preference queries to different proxies. This setting allows more flexibility but at the price of obtaining incomplete (because of delegation cycles) and even intransitive preference orders (hence the violation of a rationality constraint as studied by \cite{christoffbinary}). Notably, the authors showed that it is NP-hard to decide if an incomplete ballot obtained in this setting can be completed to obtain complete and transitive preferences while respecting the constraints induced by the delegations.\\

\textit{Are delegations in an LD setting rational?} Lastly, the work of Bloembergen et al. \cite{bloembergen2018rational} gives a different perspective on the study of voters' rationality in LD. The authors consider an LD setting where voters are connected in an SN and can only delegate to their neighbors in the network. The election is on a binary issue for which some voters should vote for the 0 answer and the others should vote for the 1 answer (voters of type $\tau_0$ or $\tau_1$). Each voter $i$ does not know exactly her type which is modeled by an accuracy $q_i\in[0.5,1]$ representing the probability with which voter $i$ makes the correct choice. Similarly, each pair $(i,j)$ of voters do not know if they are of the same type which is also modeled by a probability $p_{ij}$. Hence, a voter $i$ which has $j$ as guru has a probability that $j$ makes the correct vote (according to $i$) which is a formula including the $p_{ij}$ and $q_j$ values. The goal of each voter is to maximize the accuracy of her vote/delegation. Importantly, each voter that votes incurs a loss of satisfaction representing the work required to vote directly. This modeling leads to a class of games, called \emph{delegation games}. The authors proved the existence of pure Nash equilibria in several types of delegation games and gave upper and lower bounds on the price of anarchy, and the gain (i.e., the difference between the accuracy of the group after the delegation process and the one induced by direct voting) of such games.\\

Our approach is closest to this last work as we consider the same type of delegation games. However, our model of preferences over guru is more general as  we assume that each voter has a preference order over her possible gurus. These preference orders may be dictated by competence levels and types of voters as in \cite{bloembergen2018rational}. However, we do not make such hypothesis as the criteria to choose a delegate are numerous: geographic locality, cultural, political or religious identity, et caetera. Considering this more general framework strongly modifies the resulting delegation games.

\section{\label{sec:notations} Notations and  Settings} 
\subsection{Notations and Nash-stable delegation functions}
We denote by $\VoterSet=\{1,\ldots, n\}$ a set of voters that take part in a vote\footnote{Note that similarly to \cite{golzfluid}, we develop a setting where candidates are not mentioned. Proceeding in this way enables a general approach encapsulating different ways of specifying how candidates structure the preferences of voters over gurus.}. These voters are connected in an SN which is represented by an undirected graph $\SN = (\VoterSet,\EdgeSet)$, i.e., the vertices of the SN are the voters and $(i,j)\in \EdgeSet$ if voters $i$ and $j$ are connected in the SN. Let us denote by $\NeighbSet(i)$ the set of neighbors of voter $i$ in $\SN$. Each voter $i$ can declare intention to either vote herself, delegate to one of her neighbors $j\in \NeighbSet(i)$, or abstain. A \emph{delegation function} is a function $d:\VoterSet \rightarrow \VoterSet \cup \{0\}$ such that $d(i) = i$ if voter $i$ wants to vote, $d(i) = j \in \NeighbSet(i)$ if $i$ wants to delegate to $j$, and $d(i) = 0$ if $i$ wants to abstain.

Given a delegation function $d$, let $\Gurus(d)$ denote the \emph{set of gurus}, i.e., $\Gurus(d)=\{j \in \VoterSet \ |\ d(j) = j\}$. 
The \emph{guru of a voter $i \in \VoterSet$}, denoted by $\GuruOf(i,d)$, can be found by following the successive delegations starting from $i$. Formally, $\GuruOf(i,d)=j$ if there exists a sequence of voters $i_1,\ldots, i_{\ell}$ such that $d(i_k) = i_{k+1}$ for every $k \in \{1,\dots,\ell-1\}$, $i_1= i$, $i_{\ell} = j$ and $j \in \Gurus(d) \cup \{ 0 \}$.
However, it may happen that no such $j$ exists because the successive delegations starting from $i$ end up in a circuit, i.e., $i=i_1$ delegates to $i_2$, who delegates to $i_3$, and so on up to $i_{\ell}$ who delegates to $i_k$ with $k\in\{1,\ldots,\ell-1\}$. In this case, we consider that the $\ell$ voters abstain, as none of them takes the responsibility to vote, i.e., we set $\GuruOf(i_k,d) = 0$ for all $k\in\{1,\ldots,\ell\}$.
Such a definition of gurus allows to model the transitivity of delegations: if $d(i)$$=$$j$, $d(j)$$=$$g$, and $d(g)$$=$$g$, then the guru of $i$ will be $g$, even if $i$ has not delegated directly to $g$. 
Hence a voter $i$ can delegate directly to one of its neighbors in $\NeighbSet(i)$, but she can also delegate indirectly to another voter through a chain of delegations. 
Note that because voters can only delegate directly to their neighbors, such a chain of delegations coincides with a path in $\SN$. 

Given a voter $i$ and a delegation function $d$, we now consider how voter $i$ may change  
her delegation to get a guru different from her current guru $\GuruOf(i,d)$. She may decide to vote herself, to abstain, or to delegate to a neighbor $j$ with a different guru, and in the latter case she would get $\GuruOf(j,d)$ as a guru. 
We denote by $\AttBy(i,d) = \cup_{j \in \NeighbSet(i)} \GuruOf(j,d)$ the set of gurus of the neighbors of $i$ in $\SN$. Then the gurus that $i$ can get by deviating from $d(i)$ 
is exactly the set $\AttBy(i,d) \cup \{0,i\}$. 

\begin{example} \label{ex:reachableGurus}
Consider the SN represented in Figure \ref{SNEx1} with the delegation function $d$ defined by $d(1) = 4$, $d(2) = 5$, $d(3) = 4$, $d(4) = 4$, $d(5) = 5$, $d(6) = 6$, $d(7) = 8$, $d(8) = 9$ and $d(9) = 6$. For this example, $\Gurus(d) = \{4,5,6\}$, with  $\GuruOf(1,d) = \GuruOf(3,d) = \GuruOf(4,d) = 4$, $\GuruOf(2,d) = \GuruOf(5,d) = 5$ and $\GuruOf(6,d) = \GuruOf(7,d)= \GuruOf(8,d) = \GuruOf(9,d) = 6$. If she wants, voter 1 can change her delegation. If she delegates to voter 3, then this will not change her guru as voter 3 is delegating to guru 4. However, if she delegates to voter 2, then her new guru will be voter 5. In any case, she can also decide to modify her delegation to declare intention to vote or abstain. Note that in this example voter 1 cannot change unilaterally her delegation in order to have voter 6 as guru. Indeed, voter 6 does not belong to $\AttBy(1,d) = \{4,5\}$, as there is no delegation path from a neighbor of voter 1 to voter 6.

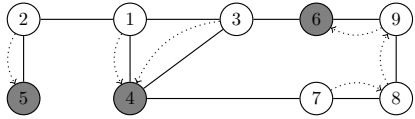
\begin{figure}[!ht] 
   \centering
    \scalebox{0.7}{\begin{tikzpicture}[-,auto,node distance=3cm,semithick]

  \node[circle,draw,text=black] (A)   at (0,0)                 {$1$};
  \node[circle,draw,text=black] (B)   at (-2,0)                 {$2$};
  \node[circle,draw,text=black] (C)   at (2,0)                 {$3$};
  \node[circle, draw,text=black, fill=gray] (D)   at (0,-1.5)                 {$4$};
  \node[circle,draw,text=black, fill=gray] (E)   at (-2,-1.5)                 {$5$};
  \node[circle,draw,text=black, fill=gray] (F)   at (3.5,0)                 {$6$};
  \node[circle,draw,text=black] (G)   at (3.5, -1.5)                 {$7$};
  \node[circle,draw,text=black] (H)   at (5, 0)                 {$9$};
  \node[circle,draw,text=black] (I)   at (5, -1.5)                 {$8$};
  
  \path (A) edge  node {} (B)
        (A) edge  node {} (C)
        (A) edge  node {} (D)
        (A) edge[bend right, dotted, ->]  node {} (D)
        (B) edge  node {} (E)
        (B) edge[bend right, dotted, ->]  node {} (E)
  	    (C) edge  node {} (D)
  	    (C) edge[bend right, dotted, ->]  node {} (D)
  	    (C) edge  node {} (F)
  	    (D) edge  node {} (G)
  	    (G) edge  node {} (I)
  	    (G) edge[bend left, dotted, ->]  node {} (I)
  	    (I) edge  node {} (H)
  	    (I) edge[bend left, dotted, ->]  node {} (H)
  	    (H) edge  node {} (F)
  	    (H) edge[bend left, dotted, ->]  node {} (F);
        
	  \end{tikzpicture}}
    \caption{SN $\SN$ in Example \ref{ex:reachableGurus}. Delegations are represented by dotted arrows and gurus are colored in gray.}
    \label{SNEx1}
\end{figure}
\end{example}

We assume that each voter $i$ has a preference order $\succ_i$ over who could be her guru in $\VoterSet\cup\{0\}$. For every voter $i \in \VoterSet$, and for every $j,k \in \VoterSet \cup \{0\}$, we have that $j \succ_i k$ 
if $i$ prefers to have $j$ as guru (or to vote herself if $j=i$, or to abstain if $j=0$) rather than to have $k$ as guru (or to vote herself if $k=i$, or to abstain if $k=0$).
The collection of preference orders 
$\{\succ_i\ |\ i \in \VoterSet\}$ in turn  defines a \emph{preference profile} $P$. To illustrate these notations, we consider in Example~\ref{ex1} a simple instance with 3 voters. Note that this instance will be reused as building block in several places in the paper under the name of \textit{3-cycle}.
\begin{example}[3-cycle]\label{ex1}
For illustration purposes, consider the following instance in which there are $n = 3$ voters connected in a complete SN and with the following preference profile $P$:
\begin{align*}
1&:2\succ_1 1 \succ_1 3 \succ_1 0\\ 
2&:3\succ_2 2 \succ_2 1 \succ_2 0\\ 
3&:1\succ_3 3 \succ_3 2 \succ_3 0 
\end{align*}
Put another way, each voter $i$ prefers to delegate to $(i \mod 3) + 1$ rather than to vote directly and each voter prefers to vote rather than to abstain.
\end{example}
As a consequence of successive delegations, a voter might end up in a situation in which she prefers to vote or to abstain or to delegate to another guru that she can reach than to maintain her current delegation.  Such a situation is regarded as {\em unstable} as this voter would modify unilaterally her delegation.  
This is for instance the case in the previous example if $d(1)=2$, $d(2)=3$ and $d(3)=3$: by successive delegations, the guru of $1$ is $3$, but $1$ would prefer to vote instead.   
More formally, a delegation function $d$ is \emph{Nash-stable for voter $i$} if 
$$\GuruOf(i,d) \succ_i g\quad \forall g \in (\AttBy(i,d) \cup \{0,i\}) \setminus \{\GuruOf(i,d)\}.$$
A delegation function $d$ is \emph{Nash-stable} if it is Nash-stable for every voter in $\VoterSet$. A Nash-stable delegation function is also called an \emph{equilibrium} in the sequel.

It may seem difficult in practice that voters give a complete linear order over the set $\VoterSet \OrAbst$. We now highlight that the computation of equilibria does not require the whole preference profile.
We say that voter $i$ is an \emph{abstainer} in $P$ if she prefers to abstain rather than to vote, i.e., if $0 \succ_i i$; she is a \emph{non-abstainer} otherwise. We will denote by $\AbstSet$ the set of abstainers. Note that an abstainer never votes directly in an equilibrium; similarly, a non-abstainer never abstains in an equilibrium. 
Given a preference profile $P$, we define $\Acc(i) = \{j \in \VoterSet | j \succ_i i \text{ and } j \succ_i 0\}$ the set of \emph{acceptable gurus} for $i$. Note that this set is likely to be quite small compared to $n$. A necessary condition for a delegation function to be Nash-stable is that $\GuruOf(i,d)\in \Acc(i)$, or $\GuruOf(i,d)=0$ and $i\in \AbstSet$, or $\GuruOf(i,d)=i$ and $i \notin \AbstSet$.
Hence when looking for equilibria, the preferences of voter $i$ below $0$ (if $i \in \AbstSet$) or $i$ (if $i \notin \AbstSet$) can be dropped. In the sequel, we may define a preference profile only by giving, for every voter $i$, if she is an abstainer or not, and her preference profile on $\Acc(i)$.

\subsection{Existence and optimization problems investigated}

Stable situations are obviously desirable.  Sadly, there are instances for which there is no equilibrium.

\begin{obs}
The instance described in Example \ref{ex1} (3-cycle) admits no  equilibrium.
\end{obs}
\begin{proof}
Assume by contradiction that there exists a Nash-stable delegation function $d$. First note that as the SN is complete, any voter can delegate to any other voter. Second, note that for any pair of voters in $\VoterSet$, there is always one voter that approves the other as possible guru, i.e., she prefers to delegate to this voter rather than to vote or abstain. Hence, $|\Gurus(d)|$ cannot be greater than 1, otherwise one of the guru would rather delegate to another guru than vote directly. On the other hand, note that there is no voter that is approved as possible guru by all other voters. Hence, $|\Gurus(d)|$ cannot be less than 2, otherwise one of the voter in $\VoterSet\backslash \Gurus(d)$ would rather vote than delegate to one of the gurus (as in this example $\AbstSet=\emptyset$). We obtain the desired contradiction.
\end{proof}

Hence the first problem, called \textbf{EXISTENCE}, that we will investigate in this article is the one of the existence of an equilibrium.\\

\noindent\fbox{\parbox{0.98\columnwidth}{
\textbf{EXISTENCE (abbreviated by EX)}\\
\emph{INSTANCE:} A preference profile $P$ and a social network $\SN$.\\
\emph{QUESTION:}  Does there exist an equilibrium?}}\\

Note that problem \textbf{EX} is in NP, as given a delegation function, we can easily find the guru of each voter and then check the Nash-stability condition in polynomial time.\\

For instances for which we know that some equilibrium exists, we will investigate if we can compute equilibria verifying particular desirable properties. 
Firstly, given a voter $i\in \VoterSet\setminus\AbstSet$, we will try to know if there exists a Nash-stable delegation function $d$ for which $i$ is a guru, i.e., $i\in \Gurus(d)$. We call this problem \textbf{MEMBERSHIP}.\\

\noindent\fbox{\parbox{0.98\columnwidth}{
\textbf{MEMBERSHIP (abbreviated by MEMB)}\\
\emph{INSTANCE:} A preference profile $P$, a social network $\SN$ and a voter $i \in \VoterSet\setminus\AbstSet$.\\
\emph{QUESTION:}  Does there exist a Nash-stable delegation function $d$ for which $i\in \Gurus(d)$?
}}\\

Secondly, we will try to find a Nash-stable delegation function that optimizes some objective function either to a) minimize the dissatisfaction of the voters, b) minimize the maximum voting power of a guru, or c) minimize the number of voters who abstain. More formally, we will study the three following optimization problems:
\begin{itemize}
    \item \textbf{MIN DISSATISFACTION} (\textbf{MINDIS} for short) minimizes the sum of dissatisfactions of the voters where the dissatisfaction of a voter $i$ is given by  $(\Score(i,d)-1)$, $\Score(i,d)$ being the rank of $\GuruOf(i,d)$ in the preference profile of $i$;
    \item \textbf{MIN MAX VOTING POWER} (\textbf{MINMAXVP} for short) minimizes the maximal voting power of a guru, where the voting power of a guru $g$ is defined by $\VotingPower(g,d)$$=$$|\{j\!\in\! \VoterSet| \GuruOf(j,d)$$=$$g\}|$;
    \item \textbf{MIN ABSTENTION} (\textbf{MINABST} for short) minimizes the number of voters that abstain in an equilibrium, i.e, it minimizes $|\{i\in \VoterSet | \GuruOf(i,d) = 0\}|$.
\end{itemize}

\noindent\fbox{\parbox{0.98\columnwidth}{
Problems \textbf{MINDIS}, \textbf{MINMAXVP} and \textbf{MINABST} \\
\emph{INSTANCE:} A preference profile $P$ and a social network $\SN$.
\\
\emph{SOLUTION:} A Nash-stable delegation function $d$.\\
\emph{MEASURE for} \textbf{MINDIS}: $\sum_{i\in \VoterSet} (\Score(i,d)-1)$ (to minimize).\\
\emph{MEASURE for} \textbf{MINMAXVP}: $\max_{i\in \Gurus(d)} \VotingPower(i,d)$ (to minimize).\\
\emph{MEASURE for} \textbf{MINABST}: $|\{i\in \VoterSet | \GuruOf(i,d) = 0\}|$ (to minimize).
}}\\

The last questions that we investigate capture the dynamic nature of delegations. 

\subsection{Convergence problems investigated}

In situations where an equilibrium exists, a natural question is whether a dynamic delegation process (necessarily) converges towards such an equilibrium. As classically done in game theory (see for instance~\cite{NisanSVZ11}), we will consider dynamics where iteratively one voter has the possibility to change her delegation\footnote{This type of delegation dynamics could be for instance implemented by using the \emph{statement voting} framework defined by Zhang and Zhou \cite{zhang2017brief} in which each voter can prescribe his own delegation rule. Each voter would then describe a statement (conditional actions) which would ensure that she is represented by her favorite attainable guru.}. In the dynamics called {\it better response dynamics}, the voter will try to improve her outcome if possible. In the {\it best response dynamics} a voter $i$ chooses $d(i)$ so as to maximize her outcome. Let us define this more formally. 

Given a delegation function $d$ and a voter $i$, let:
\begin{itemize}
\item $I_d(i)$ be the set of improved ``moves'' for $i$: $I_d(i)=\{j\in \NeighbSet(i) \cup\{0\}: \GuruOf(i,d_{i\rightarrow j}) \succ_i \GuruOf(i,d)\}$ where $d_{i\rightarrow j}$ is the same delegation as $d$ up to the fact that $i$ delegates to $j$ (or votes if $j=i$, or abstains if $j=0$). Note that we do not consider moves where voter $i$ delegates to a voter that abstains or that creates a cycle. In this case, she would rather abstain herself, i.e., set $d(i)=0$.
\item $B_d(i)$ be the set of $j\in I_d(i)\cup \{d(i)\}$ maximizing $i$'s outcome. If there is no possible improvement ($I_d(i)= \emptyset$) then $B_d(i)=\{d(i)\}$ ($i$ will not change her delegation).  
\end{itemize}

In a dynamics, we are given a starting delegation function $d_0$ and a token function $T:\mathbb{N}^* \rightarrow \VoterSet$.
The most simple choice of starting delegation function $d_0$ is the one for which every voter declares intention to vote. Note however that the convergence of a dynamics depends on $d_0$, hence we will specify in our results the starting delegation functions under consideration.

The token function specifies that voter $T(t)$ has the token at step $t$: she has the right to change her delegation. This gives a sequence of delegation functions $(d_t)_{t\in \mathbb{N}}$ where for any $t\in \mathbb{N}^*$, if $j\neq T(t)$ then $d_t(j)=d_{t-1}(j)$. A dynamics is said to converge if there is a $t^*$ such that for all $t\geq t^*$ $d_{t^*}=d_t$. 
We will assume, as usual, that each voter has the token an infinite number of times. A classical way of choosing such a function $T$ is to consider a permutation $\sigma$ over the voters in $\VoterSet$, and to repeat this permutation over the time to give the token (if $t=r\mod n$ then $T(t)=\sigma(r)$). We will call these dynamics {\it permutation dynamics}.

Given $d_0$ and $T$, a dynamics is called:
\begin{itemize}
	\item A better response dynamics or Improved Response Dynamics (IRD) if for all $t$, $T(t)$ chooses an improved move if any, otherwise does not change her delegation: if $I_{d_{t-1}}(T(t))\neq \emptyset$ then $d_t(T(t))\in I_{d_{t-1}}(T(t))$, otherwise $d_t(T(t))=d_{t-1}(T(t))$.
    \item A Best Response Dynamics (BRD) if for all $t$, $T(t)$ chooses a move in $B_{d_{t-1}}(T(t))$. 
\end{itemize}
Note that a BRD is also an IRD.

The last problems that we investigate, denoted by \textbf{IR-CONVERGENCE} (\textbf{IR-CONV} for short) and \textbf{BR-CONVERGENCE} (\textbf{BR-CONV} for short), can be formalized as:\\

\noindent\fbox{\parbox{0.98\columnwidth}{
\textbf{IR-CONV (resp. BR-CONV)}\\
\emph{INSTANCE:} A preference profile $P$ and a social network $\SN$.\\
\emph{QUESTION:}  Does a dynamic delegation process under IRD (resp. BRD) necessarily converge whatever the token function $T$?
}}

\subsection{Summary of results and outline of the paper}
Our results are presented in Table~\ref{tab:synthesis}. 
In Section~\ref{sec:hardness}, we investigate the complexity of problem \textbf{EX}. We will show that when $\SN$ is complete, this problem is equivalent to the problem of determining if a digraph admits a kernel (i.e., an independent set of nodes $S$ such that for every other node $u$ not in $S$, there exists an arc $(u,v)$ with $v\in S$) which is NP-complete \cite{chvatal1973computational}. We then strengthen this result  by showing that \textbf{EX} is also NP-complete when the maximum degree of $\SN$ is bounded by 5 and is W[1]-hard w.r.t. the treewidth of $\SN$. These results are summarized in the left tabular of Table~\ref{tab:synthesis}. Hence, deciding if an instance admits an equilibrium is an NP-complete problem. However we will identify specific SNs that ensure that an equilibrium exists whatever the preference profile $P$. More precisely, we will see that an equilibrium exists whatever the preferences of the voters iff the SN is a tree. Hence in Section~\ref{sec:trees}, we investigate the class of tree SNs and we design a dynamic programming scheme which allows to solve problems \textbf{MEMB}, \textbf{MINDIS}, \textbf{MINMAXVP} and \textbf{MINABST} in polynomial time. The polynomial complexity results we obtain are given in the central tabular of Table~\ref{tab:synthesis}. Lastly, in Section \ref{sec:conv}, we study the convergence of delegations dynamics in LD. Unfortunately, when an equilibrium exists, we show that a BRD may not converge even if $\SN$ is complete or is a path. For a star SN, we obtain that a BRD will always converge, whereas an IRD may not. These results are summarized in the right tabular of Table~\ref{tab:synthesis}.

\begin{table}[!h]
\begin{center}
\scalebox{0.76}{
\begin{tabular}{|c|c|}
\hline
Type of $\SN$ or parameter & \textbf{EX} \\
\hline
Complete & NP-C\\
\hline
Maximum degree $= 5$ & NP-C\\
\hline
Treewidth & W[1]-hard\\
\hline
Tree & AE\\
\hline
\end{tabular}}
\scalebox{0.76}{
\begin{tabular}{|c|c|c|}
\hline
Problem$\backslash$ $\SN$ & Star & Tree \\ 
\hline
\textbf{MEMB} & $O(n^2)$ & $O(n^3)$\\  
\hline
\textbf{MINDIS} & $O(n^2)$ & $O(n^3)$\\
\hline
\textbf{MINMAXVP} & $O(n^2)$ & $O(n^4)$\\ 
\hline
\textbf{ABST} & $O(n^2)$ & $O(n^3)$ \\
\hline
\end{tabular}}
\scalebox{0.76}{
\begin{tabular}{|c|c|c|c|}
\hline
$\SN$ & Problem & \textbf{IR} & \textbf{BR}\\
\cline{2-2}
\multicolumn{2}{|c|}{with an equilibrium}  & \textbf{-CONV} & \textbf{-CONV}\\
\hline
\multicolumn{2}{|c|}{Star} & NA & Always \\
\hline
\multicolumn{2}{|c|}{Path} & NA & NA \\
\hline
\multicolumn{2}{|c|}{Complete} & NA & NA \\
\hline
\end{tabular}}
\end{center}
\caption{\label{tab:synthesis} Synthesis of results (AE is for Always Exists; NA for Not Always; NP-C for NP-Complete).}
\end{table}

\section{Existence of equilibria: hardness results} \label{sec:hardness}

\subsection{Complete social networks}\label{subsec:complete}
We focus in this subsection on the case where the SN is a complete graph. We mainly show that determining whether an equilibrium exists or not is an NP-complete problem (Theorem~\ref{theo:equivKernelNashStable}), by showing an equivalence with the problem of finding a kernel in a graph. This equivalence is also helpful to find subcases where an equilibrium always exist.\\

We define the \emph{delegation-acceptability digraph} $\GP = (\VoterSet\setminus\AbstSet, \AP)$ by its arc-set $\AP = \{ (i,j) \ |\ j\in \Acc(i)\}$. Stated differently, there is one vertex per non-abstainer and there exists an arc from $i$ to $j$ if $i$ accepts $j$ as a guru. For example, in Figure \ref{delAcc}, we give a partial  preference profile $P$ involving 5 voters and the corresponding delegation-acceptability digraph $\GP$.

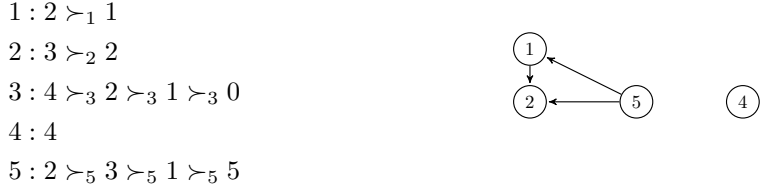
\begin{figure}[!ht] 
\begin{minipage}[c]{.46\linewidth}
\begin{align*}
1&:2\succ_1 1 \\ 
2&:3\succ_2 2 \\ 
3&:4\succ_3 2 \succ_3 1 \succ_3 0 \\
4&:4\\
5&:2\succ_5 3 \succ_5 1 \succ_5 5 
\end{align*}
\end{minipage}\hfill
\begin{minipage}[c]{.46\linewidth}
\scalebox{0.7}{\begin{tikzpicture}[->,>=stealth',shorten >=1pt,auto,node distance=3cm,semithick]
  
  \node[circle,draw,text=black] (A2)   at (0,1)                 {$1$};
  \node[circle,draw,text=black] (B2)   at (0,0)                 {$2$};
  \node[circle,draw,text=black] (D2)   at (4, 0)               {$4$};
  \node[circle,draw,text=black] (E2)   at (2, 0)                 {$5$};
  
  \path (A2) edge  node {} (B2)  
	    (E2) edge  node {} (B2)
	    (E2) edge  node {} (A2);
	  \end{tikzpicture}}
\end{minipage}
    \caption{A partial preference profile $P$ involving 5 voters (left-hand side of the figure) and the corresponding delegation-acceptability digraph  $\GP$ (right-hand side of the figure).}
    \label{delAcc}
\end{figure}

The main result of this subsection, stated in Proposition~\ref{prop:equivGurusKernel}, is a characterization of all sets of gurus of equilibria, as specific subsets of vertices of the delegation-acceptability digraph. Let us introduce additional graph-theoretic definitions.
Given a digraph $G = (\mathtt{V}, \mathtt{A})$, a subset of vertices $K \subset \mathtt{V}$ is \emph{independent} if there is no arc between two vertices of $K$. It is \emph{absorbing} if for every vertex $u \notin K$, there exists $k \in K$ s.t. $(u,k) \in \mathtt{A}$ (we say that $k$ \emph{absorbs} $u$). A \emph{kernel} of $G$ is an independent and absorbing subset of vertices. 

\begin{proposition}
\label{prop:equivGurusKernel}
Assume $\SN$ is complete, then given a preference profile $P$ and a subset of voters $K \subseteq \VoterSet$, the following propositions are equivalent:
\begin{itemize}
    \item[(i)] there exists an equilibrium $d$ s.t. $\Gurus(d) = K$;
    \item[(ii)] $K$ contains no abstainer and $K$ is a kernel of $\GP$.
\end{itemize}  
\end{proposition}

\proof
\emph{(i) $\implies$ (ii).} Let $d$ be a Nash-stable delegation function for $P$. Let us prove that its set of gurus $\Gurus(d)$ satisfy condition (ii). It was noted previously that Nash-stability implies the absence of abstainer in $\Gurus(d)$. Assume that $\Gurus(d)$ is not independent in $\GP$. Then, there exists $i,j \in \Gurus(d)$ such that $(i,j)$ is an arc of $\GP$, that is, $j \in \Acc(i)$. It implies that $i$ prefers to delegate to $j$ rather than remaining a guru. As $\SN$ is complete, $j\in \AttBy(i,d)$ and hence $d$ is not Nash-stable for $i$. Assume now that $\Gurus(d)$ is not absorbing for all non-abstainers. Then there exists a non-abstainer $i \notin \Gurus(d)$ such that for every guru $g$ in $\Gurus(d)$, $(i,g)$ is not an arc of $\GP$, that is, $g \notin \Acc(i)$. Such a voter $i$ prefers to vote herself rather than delegate to any guru in $\Gurus(d)$. Therefore $d$ is not Nash-stable for $i$.
This proves that $\Gurus(d)$ is a kernel of $\GP$.

\emph{(ii) $\implies$ (i).} Consider a subset $K$ of non-abstainers such that $K$ is a kernel of the delegation-acceptability digraph $\GP$. We define a delegation function $d$ by: $d(i) = i$ if $i \in K$; and $d(i) = j$ if $i \notin K$ where $j$ is the voter that $i$ prefers in $K \cup \{0\}$. Note that as $\SN$ is complete, each voter in $\VoterSet\setminus K$ can delegate directly to any voter in $K$. It follows that $\GuruOf(i,d) = d(i)$ for every $i$, and the set of gurus is $\Gurus(d) = K$.
Let us check that the delegation function $d$ is Nash-stable.
First, it holds that $d$ is Nash-stable for every abstainer~$i$. Indeed, the guru of $i$ is $d(i)$, which is her preferred guru in $\Gurus(d) \cup \{0\}$.
Also, we show that $d$ is Nash-stable for any $i \in \Gurus(d)$. We need to prove that $i$ does not prefer to delegate to any other guru, or to abstain. By assumption the set $\Gurus(d) = K$ contains no abstainer. Since $K$ is independent, for every other guru $g \in K$, the arc $(i,g)$ does not exist in the delegation acceptability digraph, meaning that $g \notin \Acc(i)$ and $i \succ_i g$.
Finally, we prove that $d$ is Nash-stable for any non-abstainer $i \notin \Gurus(d)$. Since $i$ has already chosen her preferred guru in $\Gurus(d)$, it is only necessary to check that $i$ does not prefer to vote herself. Because the set $K$ is absorbing, there exists $k \in K$ such that $(i,k)$ is an arc of $\GP$. Hence such $k$ is a guru in $\Gurus(d)$ such that $k \in \Acc(i)$, hence $k \succ_i i$, and $\GuruOf(i,d) \succ_i i$.
\endproof

Note that any digraph is the delegation-acceptability digraph of a preference profile $P$. Indeed, given the digraph, it suffices to build a preference profile $P$ so that every voter prefers to delegate to its out-neighbors, then to vote, then to delegate to other voters, then to abstain. Consequently, in a complete SN, determining if a preference profile admits an equilibrium is equivalent to the problem of determining if a digraph admits a kernel, which is an NP-complete problem~\cite{chvatal1973computational}. 

\begin{theorem}
\label{theo:equivKernelNashStable}
 \textbf{EX} is NP-complete even when the social network is a complete graph.
\end{theorem}

\noindent\textit{Consequences.}  
As a direct consequence of the results of this section, even when $\SN$ is complete, optimization problems \textbf{MINDIS}, \textbf{MINMAXVP} and \textbf{MINABST} are NP-hard as it is NP-hard to decide if their set of admissible solutions is empty or not. We also directly obtain that the decision problem \textbf{MEMB} is NP-complete, by a direct reduction from \textbf{EX}. Indeed, solving problem \textbf{MEMB} for each voter in $\VoterSet\setminus\AbstSet$ yields the answer to problem \textbf{EX}. 

As mentioned above, we also point out that this equivalence is useful to find some interesting subcases. Let us consider for instance the case where there is a symmetry in the preferences in the sense that $i\in \Acc(j)$ if and only if $j\in \Acc(i)$. In this case of symmetrical preference profiles, the delegation-acceptability digraph has the arc $(i,j)$ iff it has the arc $(j,i)$ (it is symmetrical). Then, any inclusion maximal independent set is a kernel. Hence, for any non-abstainer $i$ there exists  an equilibrium in which $i$ is a guru (take a maximal independent set containing $i$).

\begin{proposition}
In a complete social network, there always exists an equilibrium when preferences are symmetrical. Moreover, the answer to {\bf MEMB} is always yes.
\end{proposition}

More generally when preferences are symmetrical, given a set of non-abstainers it is easy to decide if there exists an equilibrium in which every voter in this set is a guru: we just have to check whether the set is independent or not in $\GP$.

\subsection{Sparse social networks} \label{subsec:sparse}
As the problem \textbf{EX} is NP-complete when the social network is a complete graph, it remains NP-complete in any class of graphs that contain cliques, such as interval graphs, split graphs, dense graphs,\dots In this section, we focus on classes of graphs that {\it do not} contain large cliques. We first deal with bounded degree graphs, and show that \textbf{EX} remains NP-hard in social networks of degree bounded by 5 (Theorem~\ref{th:hardnessboundeddegree}). We then focus on graphs of bounded treewidth. Interestingly, while we will see in Section~\ref{sec:trees} that \textbf{EX} is polynomial if the social network is a tree (actually, an equilibrium always exists in trees), we prove here that \textbf{EX} is W[1]-hard when parameterized by the treewidth of the social network (Theorem~\ref{th:w1hardness}).   

We refer the reader to \cite{BK08} for the standard notion of treewidth of graphs. A problem is said to be fixed parameter tractable (FPT) with respect to some parameter $k$ (the treewidth of the graph for us) if it can be solved by an algorithm whose complexity is $O(f(k)|I|^c)$ for some function $f$ and constant $c$ (where $|I|$ is the size of the instance). A W[1]-hard problem is not FPT unless FPT=W[1]. We refer the reader to~\cite{DowneyFellowsBook} for notions of parameterized complexity.

\begin{theorem}\label{th:hardnessboundeddegree}
Problem \textbf{EX} is NP-complete even if the maximum degree of the SN is at most 5.
\end{theorem}

\proof
We make a reduction from the NP-complete problem 3-SAT-4~\cite{Tovey84}, which is the restriction of the 3-SAT problem where each variable appears at most 4 times.  Two voters $v_{i}^t$ and $v_{i}^f$ are created for each truth variable $x_i$. These voters accept each other as possible gurus and are connected in the SN. Three voters are created for clause $c_j$: $v_{j1}^c, v_{j2}^c, v_{j3}^c$. These voters are connected in the SN in a 3-cycle such that $v_{ji}^c$ accepts $v_{j(i\mod{3}+1)}^c$ as a guru but rejects $v_{j(i-2\mod{3}+1)}^c$ as a guru. Moreover, $v_{j1}^c$ is connected to the three vertices $v^t_i$ or $v^f_i$ corresponding to the literals of the clause. Note that all these voters are non-abstainers. 
We illustrate the reduction by showing how a clause $c_1 = x_1 \lor x_2 \lor \lnot x_3$ is handled (see also Fig.\ref{red:3-SAT-4}):

\begin{minipage}[c]{.46\linewidth}
\begin{itemize}
    \item $\Acc(v_{11}^c) = \{v_{12}^c,v_{1}^t, v_{2}^t, v_{3}^f\}$
    \item $\Acc(v_{12}^c) = \{v_{13}^c,v_{1}^t, v_{2}^t, v_{3}^f\}$
    \item $\Acc(v_{13}^c) = \{v_{11}^c,v_{1}^t, v_{2}^t, v_{3}^f\}$
\end{itemize}
\end{minipage}\hfill
\begin{minipage}[c]{.46\linewidth}
\begin{itemize}
    \item $\Acc(v_{1}^t) = \{v_{1}^f\}$ and $\Acc(v_{1}^f) = \{v_{1}^t\}$ 
    \item $\Acc(v_{2}^t) = \{v_{2}^f\}$ and $\Acc(v_{2}^f) = \{v_{2}^t\}$ 
    \item $\Acc(v_{3}^t) = \{v_{3}^f\}$ and $\Acc(v_{3}^f) = \{v_{3}^t\}$ 
\end{itemize}
\end{minipage}

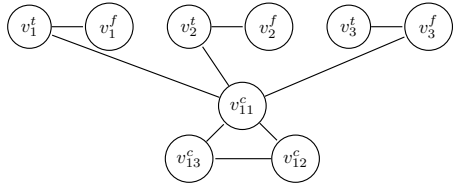
\begin{figure}[!ht] 
   \centering
   \scalebox{0.7}{\begin{tikzpicture}[-,>=stealth',shorten >=1pt,auto,node distance=3cm,semithick]

  \node[circle,draw,text=black]    (A)   at (0,0)  {$v_{11}^c$};
  \node[circle,draw,text=black]    (B)   at (1, -1)  {$v_{12}^c$};
  \node[circle,draw,text=black]    (C)   at (-1,-1) {$v_{13}^c$};
  
  \node[circle,draw,text=black] (D)   at (-4,1.5)    {$v_{1}^t$};
  \node[circle,draw,text=black] (E)   at (-1.,1.5)    {$v_{2}^t$};
  \node[circle,draw,text=black] (F)   at (2, 1.5)    {$v_{3}^t$};

  \node[circle,draw,text=black] (G)   at (-2.5,1.5)    {$v_{1}^f$};
  \node[circle,draw,text=black] (H)   at (0.5,1.5)    {$v_{2}^f$};
  \node[circle,draw,text=black] (I)   at (3.5,1.5)    {$v_{3}^f$};

  \path (A) edge  node {} (B)
  (B) edge  node {} (C)
  (C) edge  node {} (A)
  (A) edge  node {} (D)
  (A) edge  node {} (E)
  (A) edge  node {} (I)
  (D) edge  node {} (G)
  (E) edge  node {} (H)
  (F) edge  node {} (I);
        	  \end{tikzpicture}}
    \caption{\label{red:3-SAT-4}Illustration of the social network in the reduction.}
\end{figure}


The induced instance has  
maximum degree bounded by 5. We conclude the proof by proving that the 3-SAT-4 instance is satisfiable iff there exists an equilibrium.  

Suppose first that there is a truth assignment. Then consider the delegation function where: 1) a voter $v^t_i$ (resp. $v^f_i$) votes if the corresponding literal is true, otherwise she delegates to $v^f_i$ (resp. $v^t_i$); 2) a voter $v^c_{j1}$ delegates to a voter corresponding to a true literal of this clause (the best one if there are several of them); 3) both $v^c_{j2}$ and $v^c_{j3}$ delegate to $v^c_{j1}$.\\
 Voters $v^c_{j2}$ and $v^c_{j3}$ have the same guru as $v^c_{j1}$. We can easily see that this is an equilibrium.
 
 Conversely, suppose that we have an equilibrium. For each $i$, exactly one voter among $v^t_i$ and $v^f_i$ votes. We define the truth assignment where $x_i$ is true (resp. false) if $v^t_i$ votes (resp. $v^f_i$ votes). In an equilibrium, the voter $v^c_{j1}$ has to delegate to some voter $v^t_i$ or $v^f_i$ (otherwise the 3 voters of the clause $c_j$ cannot be in a stable state, as they would form a 3-cycle). Then, the guru of $v^c_{j1}$ must be a voter $v^t_i$ or $v^f_i$ that she accepts, i.e., this guru must correspond to a literal of $c_j$. 
 Hence, the truth assignment satisfies this clause.
\endproof

\begin{theorem}\label{th:w1hardness}
Problem \textbf{EX} is W[1]-hard when parameterized by the treewidth of the social network.
\end{theorem}
\proof
We make a reduction from the list coloring problem. This problem takes as input a graph $G$, together with an assignment to each vertex $v$ of a set of colors $C_v$. The problem is to determine whether it is possible to choose a color for vertex $v$ from the set of permitted colors $C_v$, for each vertex, so that the obtained coloring is proper. This problem is W[1]-hard, parameterized by the treewidth of $G$ \cite{FellowsFLRSST11}. 

To perform our reduction we show how to create the SN starting from graph $G$. The SN contains all nodes of $G$. Let $(u,v)$ be an edge of $G$. For each color $i \in C_u\cap C_v$ create three voters $d^i_1, d^i_2$ and $d^i_3$ (we omit the reference to $u,v$ to keep notation readable) such that $u$ is connected to $d^i_1$, $d^i_2$ is connected to $v$ and $d^i_1, d^i_2$ and $d^i_3$ are connected in a triangle. For each color $i \in C_u$ create a voter $c^i_u$ which is only connected to $u$ in the SN. This ends the construction of the SN. We illustrate this construction in Figure~\ref{red:tw}, representing the SN obtained for two adjacent vertices $(u,v)$ with $C_u=\{1,2,3\}$ and $C_v=\{1,3,4,5\}$. Colors $1$ and $3$ are common to $C_u$ and $C_v$, corresponding to the 2 triangles $(d^1_1,d^1_2,d^1_3)$ and $(d^3_1,d^3_2,d^3_3)$.

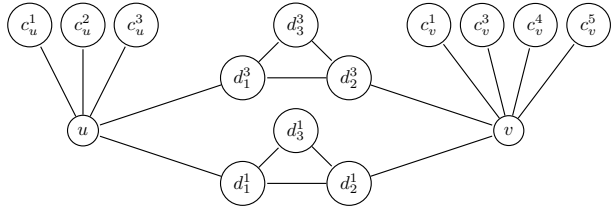
\begin{figure}[!ht] 
   \centering
    \scalebox{0.7}{\begin{tikzpicture}[-,>=stealth',shorten >=1pt,auto,node distance=3cm,semithick]

  \node[circle,draw,text=black]    (A)   at (0,0)  {$d^1_3$};
  \node[circle,draw,text=black]    (B)   at (1, -1)  {$d^1_2$};
  \node[circle,draw,text=black]    (C)   at (-1,-1) {$d^1_1$};

  \node[circle,draw,text=black]    (D)   at (0,2)  {$d^3_3$};
  \node[circle,draw,text=black]    (E)   at (1, 1)  {$d^3_2$};
  \node[circle,draw,text=black]    (F)   at (-1,1) {$d^3_1$};

  \node[circle,draw,text=black]    (G)   at (-4,0)  {$u$};
  
  \node[circle,draw,text=black] (H)   at (-5,2)    {$c^1_u$};
  \node[circle,draw,text=black] (I)   at (-4,2)    {$c^2_u$};
  \node[circle,draw,text=black] (J)   at (-3, 2)    {$c^3_u$};

  \node[circle,draw,text=black]    (K)   at (4,0)  {$v$};
  
  \node[circle,draw,text=black] (L)   at (2.5,2)    {$c^1_v$};
  \node[circle,draw,text=black] (M)   at (3.5,2)    {$c^3_v$};
  \node[circle,draw,text=black] (N)   at (4.5, 2)    {$c^4_v$};
   \node[circle,draw,text=black] (O)   at (5.5, 2)    {$c^5_v$};

  \path (A) edge  node {} (B)
  (B) edge  node {} (C)
  (C) edge  node {} (A)
  (D) edge  node {} (E)
  (E) edge  node {} (F)
  (F) edge  node {} (D)
  (C) edge  node {} (G)
  (F) edge  node {} (G)
  (G) edge  node {} (H)
  (G) edge  node {} (I)
  (G) edge  node {} (J)
  (B) edge  node {} (K)
  (E) edge  node {} (K)
  (K) edge  node {} (L)
  (K) edge  node {} (M)
  (K) edge  node {} (N)
  (K) edge  node {} (O);
        	  \end{tikzpicture}}
    \caption{\label{red:tw}Illustration of the reduction.}
\end{figure}


We now argue that the treewidth of this SN is the maximal value between 4 and the treewidth of $G$. Indeed, consider a tree decomposition $T$ of $G$. For each color $i \in C_u$, create a node that contains the bag $\{u,c^i_u\}$ and connect it to a node of $T$ that contains $u$. For each color $i \in C_u\cap C_v$ create a node that contains the bag $\{d^i_1, d^i_2, d^i_3, u, v\}$ and connect it to a node of $T$ that contains $u$ and $v$. The claim comes from the fact that we obtain in this way a tree decomposition of the SN.

We now detail the preferences of the voters. First note that there are no abstainers. Each vertex $u\in G$  accepts the vertices $c^i_u,i\in C_u$ as possible gurus.  
Each vertex $c^i_u$ accepts the vertices $c^j_u$, for $j\neq i$ in $C_u$, as possible gurus.
Let us now consider vertices $d^i_1, d^i_2$ and $d^i_3$ resulting from a color $i \in C_u\cap C_v$ for some edge $(u,v)$. 
Voter $d^i_1$ accepts $d^i_2$ as possible guru, voter $d^i_2$ accepts $d^i_3$ as possible guru and voter $d^i_3$ accepts $d^i_1$ as possible guru. Voter $d^i_1$ accepts also the voters $c^j_u$  for all $j\neq i$, and similarly voter $d^i_2$ accepts the voters $c^j_v$ for all $j\neq i$. Voter $d^i_2$ prefers any $c^j_v$ ($j\neq i$) to $d^i_3$.

We now show that the list coloring instance admits a proper coloring iff the created instance admits an equilibrium. 

Given a proper coloring of $G$ we obtain an equilibrium as follows. First, for each vertex $u$, if $u$ is colored in $G$ with color $i(u)$, then $c^{i(u)}_u$ votes and $u$ delegates to $c^{i(u)}_u$. Then, each vertex $c^i_u:i \in C_u$ with $i\neq i(u)$ delegates to $u$ in order to have $c^{i(u)}_u$ as guru. 
Let us now consider vertices $d^i_1, d^i_2$ and $d^i_3$ resulting from a color $i \in C_u\cap C_v$ for some edge $(u,v)$. If $u$ and $v$ are not assigned color $i$ ($i\neq i(u)$ and $i\neq i(v)$), then $d^i_1$ delegates to $u$, $d^i_2$ delegates to $v$ and $d^i_3$ votes. In this case, $c^{i(u)}_u$ is the guru of $d^i_1$ (note that this is an acceptable guru since $i\neq i(u)$) and $c^{i(v)}_v$ is the guru of $d^i_2$. 
Otherwise, either $u$ or $v$ is assigned color $i$ (but not both). Let us assume first that it is $u$. In this case, $d^i_1$ votes, $d^i_3$ delegates to $d^i_1$, and $d^i_2$ delegates to $v$. Voter $c^{i(v)}_v$ is then the (acceptable) guru of $d^i_2$. 
Finally, let us consider the case where $v$ is assigned color $i$. Then $d^i_1$ delegates to $u$, $d^i_3$ votes and $d^i_2$ delegates to $d^i_3$. The guru of $d^i_1$ is $c^{i(u)}_u$, the one of $d^i_2$ is $d^i_3$. We have in all cases an equilibrium.

Given an equilibrium of the instance, we first notice that for each $u$, there is exactly one vertex $c^i_u:i\in C_u$ that votes. Indeed, there cannot be more than one, because in this case, 1) $u$ would delegate to one of them and then 2) the other gurus in $c^i_u:i\in C_u$ would change to delegate to $u$ in order to have the same guru. There cannot be none: indeed voters $c^i_u:i\in C_u$ are non-abstainers, hence if they do not vote they must delegate in order to be represented by an acceptable guru. And yet the only acceptable gurus for them are the other members of $c^i_u:i\in C_u$.
Then each vertex $u$ in $G$ delegates to this particular voter $c^{i(u)}_u$. We show that coloring in $G$ each vertex $u$ with color $i(u)$ is a proper list coloring. On the contrary, suppose that two adjacent vertices $u$ and $v$ receive the same color $i=i(u)=i(v)$. We focus on the vertices $d^i_1,d^i_2,d^i_3$ (corresponding to this color $i$ for this edge $(u,v)$) in SN. $d^i_1$ cannot delegate to $u$ in the equilibrium since $c^{i}_u$ would be her guru, and she does not accept $c^{i}_u$ as a guru. Similarly, $d^i_2$ cannot delegate to $v$ in the equilibrium. Then, $d^i_1,d^i_2,d^i_3$ form a 3-cycle which cannot reach a stable state, contradiction.
\endproof

Another parameter that is worth being considered is the maximal cardinal of a set of acceptable gurus, where the maximum is taken over all voters: $\mathtt{maxa} = \max_{i\in \VoterSet} |\Acc(i)|$. This number is likely to be small in practice. Would this assumption help for solving  problem \textbf{EX}? 
Unfortunately, in the proof of Theorem~\ref{th:hardnessboundeddegree}, $\mathtt{maxa}$ is bounded above by 4, so the problem remains NP-hard when both $\mathtt{maxa}$ and the maximum degree are bounded. An interesting question would be to determine whether the problem becomes FPT when parameterized by the treewidth and $\mathtt{maxa}$. We leave this as an open question.

\section{Algorithms on Tree Social Networks}\label{sec:trees}

\subsection{Equilibria and trees}

In this section, we first answer the question of characterizing social networks in which an equilibrium always exists. It turns out that such social networks are exactly trees.

\begin{theorem} \label{thrm:treeExistence}
If $\SN$ is a tree, then for any preference profile there exists an equilibrium.
\end{theorem}

\begin{proof}
We proceed by induction. The case of a tree of 1 voter is trivial. Consider the result true up to $k$ voters and consider a tree of $k+1$ voters. Root the tree at some arbitrary vertex. Consider $l$ a leaf  
and $p$ the parent of this leaf. First if $l$ is an abstainer we build an equilibrium in the following way. There exists an equilibrium in the tree without $l$, add $l$ to this equilibrium by giving her her preferred option between abstaining and delegating to the guru of $p$. 

Now, we assume $l$ is not an abstainer. If $l\not \in \Acc(p)$ it is easy to construct an equilibrium. Again, there exists an equilibrium in the tree without $l$, add $l$ to this equilibrium by giving her her preferred option between voting and delegating to the guru of $p$. 

We now assume $l$ is not an abstainer and $l \in \Acc(p)$. In this case, we will assume that $p$ can always delegate to $l$ as last resort. Hence, we can consider that the acceptability set of $p$ can be restrained to the voters at least as preferred as $l$. Secondly, this assumption means that other voters cannot hope to have $p$ as guru. More precisely, the only gurus that they can reach through $p$ are the voters at least as preferred as $l$, and of course $l$.

To materialize these constraints, we consider the tree without $l$ and where $p$ takes the place of $l$ in all voters preference list (including the one of $p$). There exists an equilibrium in this tree. If $p$ delegates in the equilibrium (then she prefers her guru to $l$ otherwise $p$ would vote), add $l$ to this equilibrium by giving her her preferred option between voting and delegating to the guru of $p$. Otherwise, make $p$ delegate to $l$ and $l$ votes.  
\end{proof}

The other direction is also true, as shown in the following theorem.

\begin{theorem}
If $\SN$ is a social network such that for any preference profile there exists an equilibrium, then $\SN$ is a tree.
\end{theorem}
\begin{proof}
Suppose that $\SN$ is not a tree, and let us consider a chordless cycle $C=(i_1,i_2,\dots,i_k,i_1)$ (with $k\geq 3$). We consider the following preferences, with no abstainer:
\begin{itemize}
    \item A voter not in $C$ prefers to vote.
    \item Voter $i_j\in C$ has the following preference: $\Acc(i_j)$ consists of all the voters of $C$ up to $i_{j-1}$ (and $i_j$), and we have: $i_{j+1} \succ_{i_j} i_{j+2}    \dots \succ_{i_j} i_{j-3}  \succ_{i_j} i_{j-2} \succ_{i_j} i_j$.
\end{itemize}
Suppose that there is an equilibrium. Each voter not in $C$ votes, and no voter in $C$ delegates to a voter not in $C$. We focus on voters in $C$. Nobody abstains since there is no abstainer. If nobody votes, then there is a delegation cycle, so a subset of voters abstain, which is impossible in an equilibrium (they would rather vote). So let us consider one voter in $C$, say $i_1$, that declares intention to vote. Voter $i_2$ does not vote (otherwise $i_1$ would delegate to $i_2$), does not delegate to $i_1$, so she delegates to $i_3$. Similarly $i_3$ does not vote (otherwise $i_1$ would delegate to $i_2$ so that $i_3$ would be her guru), and does not delegate to $i_2$, so she delegates to $i_4$. By an easy recurrence, we get that $i_j$ delegates to $i_{j+1}$ for $j=2,\dots,k-1$. Now since $i_1$ votes, $i_k$ delegates to her. But in this case $i_1$ would be the guru of everyone, including $i_2$ who does not approve her, so $i_2$ would vote, a contradiction.
\end{proof}

\subsection{Solving optimization problems in trees}

Let us address the complexity of the problems \textbf{MEMB}, \textbf{MINDIS}, \textbf{MINMAXVP} and \textbf{MINABST} in a tree. It will be shown that when the SN is a tree, a dynamic programming approach is successful in building and optimizing equilibria.

We introduce some additional tools. Let the social network $\SN$ be a tree $\Tree$. Assume that $\Tree$ is rooted at some vertex $r_0$, and let us define accordingly, for every vertex $i \in \VoterSet$, $\parent(i)$ the parent of $i$ and $\Child(i)$ the set of its children. Let $\Subtree{i}$ denote the subtree of $\Tree$ rooted at $i$.

Let $r \!\in\! \VoterSet$ and let $d$ be an equilibrium such that the guru of $r$ is some $g_r \!\in\! \VoterSet\! \cup\!\{0\}$. Note that $g_r$ can be in $\Tr \!\OrAbst$: then $r$ delegates downwards, votes or abstains, i.e. $d(r) \!\in\! \Tr\! \cup\! \{0\}$;  or $g_r \!\notin\! \Tr\! \OrAbst$: then $r$ delegates to her parent, i.e., $d(r)\!=\!p(r)$. Consider the chain of delegations starting from a voter $j \in \Trsansr$: either it reaches $r$, and then the guru of $j$ will be $g_r$, or it does not reach $r$, hence it is fully determined by the restriction of $d$ to the subtree $\Tr$.
Using this remark we will show that an equilibrium can be built inductively by combining delegation functions in subtrees.

Let us now define local equilibria to formalize restrictions of equilibria in a subtree.
Let $d : \Tr \longrightarrow \VoterSet \OrAbst$ be a delegation function over $\Tr$.
We define gurus associated with $d$ in a similar way as we defined gurus for delegation function over all voters.
Let the guru $\GuruOf(r,d)$ of $r$ be: the first voter $i$ such that $d(i)=i$ reached by the chain of delegations starting from $r$ if $r$ delegates downwards; $0$ if $d(r)=0$; or some $g_r \notin \Tr \OrAbst$ if $d(r) = p(r)$.
Given a voter $j \in \Trsansr$, let the guru $\GuruOf(j,d)$ of $j$ be: the first voter $i$ such that $d(i)=i$ reached by the chain of delegations starting from $j$ if the chain does not reach $r$; $0$ if $d(j)=0$; or $\GuruOf(i,d) = \GuruOf(r,d)$ otherwise.

Given $d : \Tr \longrightarrow \VoterSet \OrAbst$ and $g_r \in \VoterSet \OrAbst$, we say that $d$ is a \emph{local equilibrium on $\Tr$ with label $g_r$} if it satisfies : 
\begin{itemize}
    \item[(i)] either $g_r \!\notin \Tr \!\OrAbst$ and $d(r)\! =\! p(r)$, or there is a chain of delegations in $d$ going from $r$ to $g_r$.
    \item[(ii)] for every $i \!\in\! \Tr \!\setminus\! \{r\}$, voter $i$ does not want to change her delegation, given that the guru of $r$ is $g_r$, i.e., with $\GuruOf(r,d) := g_r$ it holds that $\GuruOf(i,d) \succ_i g\ \forall g \in (\cup_{j \in \NeighbSet(i)} \GuruOf(j,d) \cup \{0,i\}) \setminus \{ \GuruOf(i,d) \}$.
    \item[(iii)] the root $r$ does not want to change her delegation to any of her children, or to vote or to abstain, given that her current guru is $g_r$, i.e., $g_r \succ_r g\ \forall g \in (\cup_{j \in \NeighbSet(i), j \not= p(r)} \GuruOf(j,d) \cup \{0,r\}) \setminus \{ g_r \}$.
\end{itemize}
Note that condition (i) means that the label is consistent with the delegations, i.e., it is indeed possible that $g_r$ is the guru of $r$ in an equilibrium that coincides with $d$ on $\Tr$. Condition (ii) corresponds to Nash-stability for voters in $\Trsansr$, and condition (iii) is a relaxed Nash-stability for $r$. 
This definition slightly generalizes the definition of equilibrium: an equilibrium $d$ is exactly a local equilibrium on $\Tree = \Subtree{r_0}$ with label $\GuruOf(r_0,d)$.


\begin{proposition}
\label{prop:dynprogtree}
Let $r \in \VoterSet$, $g_r \in \VoterSet \cup\{0\}$ and $d : \Tr \longrightarrow \VoterSet \cup\{0\}$.

Then $d$ is a local equilibrium on $\Tr$ with label $g_r$ if and only if the following assertions are satisfied:
\begin{itemize}
    \item[(a)] $d(r) = r$ (resp. $0$, $p(r)$) if $g_r = r$ (resp. $g_r= 0$, $g_r \notin \Tr \OrAbst$) and $d(r) = u^*$ if $g_r \in \Subtree{u^*}$ for some $u^* \in \Child(r)$;
    \item[(b)] $g_r \succ_r g$ for every $g \in \{ 0, r\} \setminus \{g_r\}$;
    \item[(c)] For every $u \in \Child(r)$, (c1) or (c2) is satisfied:
    \begin{itemize}
    \item[(c1)] $g_r \notin \Tu$ and there exists $g_u \in \Tu \cup \{0\}$ such that $g_r \succ_r g_u$, $g_u \succ_u g_r$, and $d$ is a local equilibrium on $\Subtree{u}$ with label $g_u$;
    \item[(c2)] $d$ is a local equilibrium on $\Subtree{u}$ with label $g_r$.
    \end{itemize}
\end{itemize}
\end{proposition}

\begin{proof}
Let $d$ be a local equilibrium on $\Tr$ with label $g_r$. It is clear that (i) implies (a), and (iii) implies (b).
Let us check (c) for a given $u \in \Child(r)$. Assume first $\GuruOf(u,d) = g_r$. Then $d$ is a local equilibrium on $\Tu$ with label $g_r$ and (c2) is satisfied.
Assume now that $\GuruOf(u,d) \not = g_r$, hence $\GuruOf(u,d)$ is some $g_u \in \Tu \OrAbst$ with $g_u \not= g_r$. It comes $g_r \notin \Tu$ since otherwise $u$ would also have $g_r$ as guru because the only path from $r$ to $g_r$ goes through $u$. Also $d$ is a local equilibrium on $\Tu$ with label $g_u$. Finally since $d$ satisfies (ii) and (iii), neither $r$ nor $u$ wants to change for the guru of each other, i.e., $g_r \succ_r g_u$ and $g_u \succ_u g_r$. Hence (c1) is satisfied. 

Conversely, let us assume that $d$ satisfies (a), (b), and (c). First condition (i) is satisfied. Indeed it is implied by (a) when $g_r \notin \Tr$ or $g_r = 0$ or $g_r = r$. Otherwise $g_r \in \Tu$ for some $u \in \Child(r)$, and since $d$ is a local equilibrium on $\Tu$ with label $g_r$, there exists a chain of delegations from $u$ to $g_r$ in $d$. By (a) $d(r)=u$ hence there exists a chain of delegations from $r$ to $g_r$. Condition (ii) is satisfied for every $j \in \Tr \setminus \Child(r)$ because $d$ is a local equilibrium in every $\Tu$. Let us check (ii) for $u \in \Child(r)$: it is sufficient to check that $\GuruOf(u,d) \succ_u \GuruOf(r,d) = g_r$ whenever $\GuruOf(u,d) \not= g_r$. This holds since $g_u \succ_u g_r$ in case (c1), and $\GuruOf(u,d) = g_r$ in case (c2). Finally let us prove that condition (iii) is satisfied. With (b) it is sufficient to check that $r$ prefers $g_r$ to any of the gurus of her children, and indeed for every $u \in \Child(r)$ by (c) it comes: either the guru of $u$ is $g_r$, or it is some $g_u$ such that $g_r \succ_r g_u$.
\end{proof}

To solve optimization problems on trees, we now present a dynamic programming approach. It is presented in a general context where the function to minimize over equilibria is a function $\Phi$ with appropriate properties. Then it is shown that criteria of the problems \textbf{MEMB}, \textbf{MINABST} and \textbf{MINDIS} can be written under that form, thus allowing us to conclude on the complexity of those three problems.
For every $i \in \VoterSet$, let $\Phi_i : g \in \VoterSet \OrAbst \longrightarrow \mathbb{R}$ be a function, and let $\Phi$ be a function defined on equilibria by $\Phi(d) = \sum_{i \in \VoterSet} \Phi_i(\GuruOf(i,d))$.
Consider the problem \textbf{MINPHI} of finding an equilibrium $d$ that minimizes $\Phi(d)$. Using Proposition~\ref{prop:dynprogtree}, one can show that this problem can be solved in polynomial time through dynamic programming.

\begin{theorem}
\label{thm:minphi}
On a tree SN, the problem \textbf{MINPHI} can be solved in $O(n^3)$. 
\end{theorem}

\begin{proof}
Given $r\! \in\! \VoterSet$ and $d$ a local equilibrium on subtree $\Tr$, we extend naturally the definition of $\Phi$ by setting $\Phi(d)\! =\! \sum_{i \in \Tr} \Phi_i(\GuruOf(i,d))$.
Let $V(r,g_r)$ be the minimum of $\Phi(d)$ for $d$ a local equilibrium on $\Tr$ with label $g_r$.
Let $\Label_u(g_r) \!=\! \{g_u\! \in\! \Tu \!\OrAbst \!\cup\! \{g_r\}$ such that $(c1)$ or $(c2) \}$. Using Prop.~\ref{prop:dynprogtree}, we establish the DP equation as follows. First $V(r,g_r) = +\infty$ if $g_r$ does not satisfies (b). Otherwise,
when computing $V(r, g_r)$, the guru of $r$ is fixed to $g_r$, thus making independent the subproblems on the children of $r$. The guru of each child $u$ of $r$ is either $g_r$ or inside $\Tu \!\OrAbst$ (more precisely in $\Label_u(g_r)$); by additivity, we just have to pick the best solution for each child and sum them up. Finally it comes
$$ V(r, g_r) = \Phi_r(g_r) + \sum_{u \in \Child(r)} \min_{g_u \in \Label_u(g_r)} V(u,g_u)$$
If $r$ is a leaf, then $V(r,g_r)$ equals $\infty$ if $g_r$ does not satisfy (b), else it equals $\Phi_r(g_r)$.
The optimal value of \textbf{MINPHI} is $\min_{g_r \in \VoterSet \OrAbst} V(r,g_r)$.
Using the DP equation, each $V(r,g_r)$ can be computed in $O(\sum_{u\in \Child(r)} |\Label_u(g_r)|)$.  We have $\sum_{u \in \Child(r)} |\Label_u(g_r)| \leq \sum_{u \in \Child(r)} (|\Tu| + 2) \leq 3|\Tr|$. Thus each $V(r,g_r)$ can be computed in $O(|\Tr|)$, and the problem \textbf{MINPHI} can be solved in $O(n^3)$.
\end{proof}

\begin{corollary}
The problems \textbf{MEMB}, \textbf{MINABST}, and \textbf{MINDIS} are solvable in $O(n^3)$.
\end{corollary}

\begin{proof}
Let us show that those three problems can be written as a \textbf{MINPHI} problem with appropriate function $\Phi$. The result will then follow from Theorem~ \ref{thm:minphi}.
Problem \textbf{MINABST} is \textbf{MINPHI} with $\Phi_i(g) = \mathbf{1}_{\{ g=0 \}}$.
Problem \textbf{MINDIS} is \textbf{MINPHI} with $\Phi_i(g)$ the rank of $g$ in $\succ_i$.
For problem \textbf{MEMB}, let $r_0$ be the voter given as input, and let $\Phi_{r_0}(r_0) = -1$, $\Phi_i(g)=0$ for every $i\not= r_0$ or $g \not= r_0$. Then the optimum of \textbf{MINPHI} is -1 iff there exists an equilibrium where $r_0$ is a guru.
\end{proof}

In contrast, the criterion of \textbf{MINMAXVP} cannot be written as a $\Phi$ function of that form; let us present a dedicated DP scheme for solving \textbf{MINMAXVP}.

\begin{theorem}
The problem \textbf{MINMAXVP} is solvable in $O(n^4)$.
\end{theorem}

\begin{proof}
Let $r \in \VoterSet$ and $g_r \in \VoterSet \OrAbst$. We first note that in a local equilibrium on $\Tr$ with label $g_r$, the voting power of all gurus in $\Tr \setminus \{g_r\}$ is fully determined by the delegation function over $\Tr$: indeed, the voters outside of $\Tr$ may only reach them 
through $r$.


Given $K \leq n$ an integer and $g_r$ in $\VoterSet$, let us define $V^K(r,g_r)$ as the minimum of $|\{ i \in \Tr \ |\ \GuruOf(i,d)=g_r \}|$, subject to: $d$ is a local equilibrium on $\Tr$ with label $g_r$, and every guru in $\Tr \setminus \{g_r\}$ has voting power at most $K$. 
Furthermore, since the abstention rate does not matter in \textbf{MINMAXVP}, we set $V^K(r,0)=0$ (or infinity if there is no local equilibrium where every guru has voting power at most $K$). 
Note that there exists an equilibrium where the max voting power is $\leq K$ if and only if $\min_{g_{r_0} \in \VoterSet \OrAbst} V^K(r_0,g_{r_0}) \leq K$.

Let us fix $K$ and show that all values $V^K(r,g_r)$ can be computed through dynamic programming, thanks to Proposition~\ref{prop:dynprogtree}, as follows. 

\begin{itemize}
    \item If $g_r$ does not satisfy (b) then $V^K(r,g_r) = +\infty$.
    \item Otherwise, for every $u \in \Child(r)$ with a label $g_u \in \Label_u(g_r)$, one can see that the subtree $\Tu$ contributes to the objective function of $V^K(r,g_r)$ in the following manner. If $g_u = g_r$ then the subtree $\Tu$ yields at least $V^K(u,g_u)$ units of voting power to $g_r$; if $g_u \not= g_r$ then it is necessary that $V^K(u,g_u) \leq K$ (otherwise $u$ would have too much voting power), and the subtree $\Tu$ does not contribute to any increase of the voting power of $g_r$. This leads to the DP equation:
$$ V^K(r, g_r) = \mathbf{1}_{g_r \neq 0} + \sum_{u \in \Child(r)} \min_{g_u \in \Label_u(g_r)} \begin{cases}
      V^K(u,g_u) & \text{if}\ g_u = g_r \\
      0 & \text{if}\ g_u \not= g_r\text{ and }V^K(u,g_u)\leq K \\
      +\infty & \text{otherwise} \\
    \end{cases}
$$
\end{itemize}
If $r$ is a leaf, then $V^K(r,g_r)$ equals $\infty$ if $g_r$ does not satisfy (b), else it equals 0 if $g_r=0$ and 1 otherwise.

Thus each $V^K(r,g_r)$ can be computed in $O(\sum_{u \in \Child(r)}|\Label_u(g_r)|) = O(|\Tr|)$. The complete table $V^K(\cdot, \cdot)$ can be computed in $O(n^3)$. Finally the problem \textbf{MINMAXVP} can be solved in $O(n^4)$.
\end{proof}

\paragraph{The Subcase of Star Social Networks.} Interestingly, in the case of a star SN (which we assume rooted in its center $c$), our dynamic programming approaches can be used to solve problems \textbf{MEMB}, \textbf{MINABST}, \textbf{MINDIS} and \textbf{MINMAXVP} in $O(n^2)$. Indeed, note that for any leaf $l$ of the tree, any value $V(l,g_l)$ can be computed in constant time. Hence, in the case of star SNs,
we only have to compute $O(n)$ values in $O(n)$ operations (the values $V(c,g_c)$, where $c$ is the center of the star), and the other values can be computed in constant time. Hence the $O(n^2)$ complexity to solve problems \textbf{MEMB}, \textbf{MINABST} and \textbf{MINDIS}. For problem \textbf{MINMAXVP}, observe that the same argument holds to compute all values $V^K(r,g_r)$ in $O(n^2)$ operations for any $K\le n$. We conclude the justification of our claim by noting that in the case of a star rooted in its center, it is sufficient to compute these values for $K\le 1$  because it will always hold that a guru in $\mathcal{T}_r\setminus\{g_r\}$ has voting power at most 1.

\section{Convergence of delegation dynamics}\label{sec:conv}
In this section, we investigate the convergence of delegation dynamics in several types of graphs. While complete SNs are investigated in Section \ref{sec:convComplete}, our results on two subclasses of tree SNs, paths and stars, are presented in Section \ref{sec:convPathAndStar}. 

\subsection{Convergence of delegation dynamics in complete social networks}\label{sec:convComplete}
As an equilibrium may not exist in the case of a complete SN (cf. Example \ref{ex1}), it is obvious that IRD or BRD do not always converge. One may wonder whether the convergence is guaranteed in instances where an equilibrium exists.  
Note that this would not contradict NP-completeness of problem \textbf{EX} in complete SNs (Theorem \ref{theo:equivKernelNashStable}) since the convergence might need an exponential number of steps. 
Unfortunately, the answer is negative.

\begin{theorem}\label{thrm:dynamics}
There exists an instance $(P,\SN)$, where $\SN$ is a complete SN and where an equilibrium exists, and a BRD that do not converge. This even holds with a permutation dynamics with a starting delegation function $d_{0}$ in which every voter declares intention to vote.
\end{theorem} 
\proof
We consider the following instance with 4 
voters 
and the following preferences:
\begin{align*}
1&: 2 \succ_1 1 \succ_1 3 \succ_1 4 \succ_1 0 \\ 
2&: 3 \succ_2 4 \succ_2 2 \succ_2 1 \succ_2 0 \\ 
3&: 2 \succ_3 1 \succ_3 3 \succ_3 4 \succ_3 0\\
4&: 3 \succ_4 4 \succ_4 2 \succ_4 1 \succ_4 0
\end{align*}
The social network is the complete graph. Note that this instance admits the following Nash-stable delegation function: \(d(1)=1, d(2)=4, d(3)=1\) and \(d(4)=4\). 

We consider the BRD permutation dynamics with function $T$ associated to the permutation $\sigma=\{4,3,1,2\}$, and we start at $d_0$ where every voter declares intention to vote.

\begin{itemize}
\item Round 1: $d_1(4)$$=$$3$, $d_2(3)$$=$$2$, $d_3(1)$$=$$2$, $d_4(2)$$=$$2$. At the end of this round, $2$ is the unique guru.
\item Round 2: $d_5(4)$$=$$4$, $d_6(3)$$=$$2$, $d_7(1)$$=$$2$, $d_8(2)$$=$$4$. At the end of this round, $4$ is the unique guru.
\item Round 3: $d_9(4)$$=$$4$, $d_{10}(3)$$=$$3$, $d_{11}(1)$$=$$1$, $d_{12}(2)$$=$$3$. At the end of this round, $1$, $3$ and $4$ are gurus.
\item Round 4: $d_{13}(4)$$=$$3$, $d_{14}(3)$$=$$1$, $d_{15}(1)$$=$$1$, $d_{16}(2)$$=$$2$. At the end of this round, $1$ and $2$ are gurus.
\item Round 5: $d_{17}(4)$$=$$4$, $d_{18}(3)$$=$$2$, $d_{19}(1)$$=$$2$, $d_{20}(2)$$=$$4$. At the end of this round, the delegation function is exactly the same as the one at the end of round 2. 
\end{itemize}
\vspace{-0.5cm}
\endproof

Hence the convergence of a BRD is not guaranteed even in instances where an equilibrium exists. In fact, the game is not even weakly acyclic, as shown in the following Theorem.
\begin{theorem}
 There exists an instance with a complete SN where an equilibrium exists, but there exists a starting delegation function from which no IRD converges to equilibrium.
\end{theorem} 
\begin{proof}
Consider an instance with 5 voters connected in a complete social network such that voters \(1\) and \(2\) approve each other as possible guru and then prefer to vote. Voters \(3\), \(4\) and \(5\) also accept voter \(1\) as guru, 
which is their preferred guru, but not voter \(2\). Then their preferences form a 3-cycle, i.e., voter \(3\) accepts voter \(4\) and then prefers to vote, voter \(4\) accepts voter \(5\) and then prefers to vote, voter \(5\) accepts voter \(3\) and then prefers to vote. The delegation function where voter \(1\) votes and all voters delegate to her is the only equilibrium of the instance. And yet if the starting delegation function is \(d_0(1) = 2\) and \(d_0(i) = i\) for all other \(i\), then there is no IRD that will converge as voter 1 will never change her delegation. 
\end{proof}

However, we note that given an equilibrium (when one exists), we can design a permutation dynamics that converges towards this equilibrium. 

\begin{theorem}
In a complete SN, for any preference profile, if an equilibrium exists, then we can find a permutation $\sigma$ inducing a permutation dynamics which starts when all voters vote, and which always converges under BRD to this  equilibrium.
\end{theorem}
\proof
Let $d$ be a Nash-stable delegation function and consider any permutation $\sigma$ such that voters in $\Gurus(d)$ are placed in the $|\Gurus(d)|$ last positions. Then, under BRD, after the $|\VoterSet \setminus\Gurus(d)|$ first delegations steps of the delegation process, no voter in $\VoterSet \setminus \Gurus(d)$ has chosen to vote. Indeed, a choice to vote by one of these voters would contradict the fact that $\Gurus(d)$ is absorbing in the delegation-acceptability digraph $\GP$. Hence, they all have decided to abstain or to delegate and are no more available as possible guru. Furthermore, as $\Gurus(d)$ is independent in $\GP$ and contains no abstainers, in the $|\Gurus(d)|$ next steps, all voters in $\Gurus(d)$ decide to vote. The second time each voters in $\VoterSet\setminus\Gurus(d)$ have the token, they choose to delegate to their favorite guru in $\Gurus(d)\cup\{0\}$ (if it is not already the case with their current delegation), and each voter $\Gurus(d)$ remains a guru. At this point, the delegation process has converged to $d$.
\endproof

\subsection{Convergence of delegation dynamics in tree social networks} \label{sec:convPathAndStar}
We now consider the convergence of delegation dynamics in instances with SNs that belong to two subclasses of tree SNs, namely path and star SNs. For these instances, an equilibrium is guaranteed to exist by Theorem \ref{thrm:treeExistence}. Unfortunately, we show that a BRD may not converge if $\SN$ is a path. Differently, if $\SN$ is a star, then BRD are guaranteed to converge. However, IRD are not. These results are presented in the two following theorems. 
\begin{theorem}
There exists a path SN and a BRD that does not converge even if there are no abstainers and if in the starting delegation function, every voter declares intention to vote. 
\end{theorem}
\begin{proof}
Consider a path of 5 voters, with the following preferences:
\begin{align*}
    1&: 5 \succ_1 1\\
    2&: 5 \succ_2 1 \succ_2 2\\
    3&: 4 \succ_3 2 \succ_3 3\\
    4&: 1 \succ_4 5 \succ_4 4\\
    5&: 1 \succ_5 5
\end{align*}
Now consider the token function $T$ and the BRD such that: $d_1(3) = 4$, $d_2(4)=5$, $d_3(5) = 5$, $d_4(2) = 3$ so her guru is $5$, $d_5(1) = 2$ so her guru is $5$. At this point every voter is delegating directly or indirectly to voter 5. Now the dynamics continues in the following way: $d_6(3)=3$, $d_7(2)=2$, $d_8(1)=1$, $d_9(3)=2$, $d_{10}(2)=1$, $d_{11}(4)=3$ so her guru is $1$, $d_{12}(5)=4$ so her guru is $1$. At this point every voter is delegating directly or indirectly to voter 1 which is the symmetric situation to the one where 5 was the guru of every voter. Now the dynamics continues in the following way: $d_{13}(3)=3$, $d_{14}(4)=4$, $d_{15}(5)=5$, $d_{16}(3)=4$, $d_{17}(4)=5$, $d_{18}(2)=3$ so her guru is $5$, $d_{19}(1)=2$ so her guru is $5$.
Now we came back to the exact same situation where 5 is the guru of all voters, which proves that by continuing in this way the BRD will loop infinitely.
\end{proof}

\begin{theorem}
If $\SN$ is a star and there are no abstainers, an IRD will always converge. If $\SN$ is a star, a BRD will always converge but there exists an IRD that do not converge. 
\end{theorem}
\begin{proof}
Consider the case where there are no abstainers. In this case, it is easy to realize that, in an IRD, each time $c$ changes her delegation, she increases her satisfaction. Hence, $c$ will change her delegation at most $n$ times. Consider the last time $c$ changes her delegation. Now each leaf can vote, delegate to $c$ or abstain, and these choice will be the same each time she gets the token. Hence, after each leaf get the token two times more, we reach an equilibrium. 

Now consider the case with possibly abstainers and a BRD. After each voter has received the token, no abstainer (rep. non-abstainer) is declaring intention to vote (resp. abstain) and they will no longer be as they prefer abstaining (resp. voting) to voting (resp. abstaining). Now, each time $c$ changes her delegation, she increases her satisfaction. By similar arguments as in the case with no abstainers, we will reach an equilibrium. 
Lastly, consider a star SN with three voters $l$, $c$ and $r$  such that $\EdgeSet = \{(l,c),(c,r)\}$. Voter $l$ (resp. $r$) first prefers to have $c$ as guru, then to abstain, then to vote, then to have $r$ (resp. $l$) as guru. Voter $c$ prefers to have $r$ as guru, then $l$, then to vote. Now consider the starting delegation function in which all voters declare intention to vote and the following token function and IRD: 
$d_1(l) = c$, $d_2(c) = r$, $d_3(l) = l$, $d_4(r)=0$, $d_5(c)=c$, $d_6(r) = c$, $d_7(c)=l$, $d_8(r)=r$, $d_9(l)=0$, $d_{10}(c)=c$, $d_{11}(l)=c$. We have came back to $d_1$ which proves the claim. 
\end{proof}

\section{Conclusion and Future Works}
\label{sec:conclusion}
We have proposed a game-theoretic model of the delegation process induced by the liquid democracy paradigm when implemented in a social network. This model makes it possible to investigate several questions on the Nash-equilibria that may be reached by the delegation process of liquid democracy. We have defined and studied several existence and optimization problems defined on these equilibria. Unfortunately, the existence of a Nash-equilibrium is not guaranteed and is even NP-Hard to decide even when the social network is complete or has a low maximum degree. In fact, a Nash-equilibrium is only guaranteed to exist whatever the preferences of the voters if the social network is a tree. Hence, we have investigated the case of tree social networks and designed efficient optimization procedures for this special case.
We give possible directions for future works. Firstly, similarly to Bloembergen et al. \cite{bloembergen2018rational} who studied the price of anarchy of delegation games with a more specific type of preferences (in \cite{bloembergen2018rational} the preferences are structured by some parameters measuring the expertise and the similarity of the voters), it would be interesting to study the price of anarchy of the delegation games developed in this paper.  
A second direction would be to investigate a similar game-theoretic analysis of delegations in other frameworks related to liquid democracy such as \emph{viscous democracy} or \emph{flexible representative democracy} \cite{boldi2011viscous,abramowitz2018flexible}. 

\bibliographystyle{alpha}
\bibliography{ourbib}

\end{document}